\documentclass{amsart}
\usepackage{amsmath, amsfonts, amsthm, amssymb}
\usepackage{paralist}
\usepackage{graphics} %% add this and next lines if pictures should be in esp format
\usepackage{epsfig} %For pictures: screened artwork should be set up with an 85 or 100 line screen
\usepackage{graphicx}  
\usepackage{multicol}
\usepackage{epstopdf}
\usepackage{algorithm}%
\usepackage{algorithmicx}%
\usepackage{algpseudocode}%

\usepackage{url}
\usepackage{pdfpages}
\usepackage[foot]{amsaddr}

\usepackage{tikz} 
\usetikzlibrary{shapes,arrows, patterns,shapes.geometric}

\newtheorem{theorem}{Theorem}[section]
\newtheorem{corollary}[theorem]{Corollary}
\newtheorem{lemma}[theorem]{Lemma}

\newtheorem{assumption}[theorem]{Assumption}

\newtheorem{proposition}[theorem]{Proposition}

\newtheorem{definition}[theorem]{Definition}

\newtheorem{example}[theorem]{Example}

\newtheorem{remark}[theorem]{Remark}

\tikzstyle{b} = [draw, thick, black, -]

\newcommand{\F}{\mathbb{F}}

\newcommand{\N}{\mathbb{N}}

\newcommand{\row}{\text{Row}}
\newcommand{\col}{\text{Col}}
\newcommand{\spec}{\text{Spec}}
\newcommand{\supp}{\text{Supp}}
\newcommand{\rmv}[1]{}

\title[Avoiding Weak Keys in the BIKE Cryptosystem]{A Combinatorial Approach to Avoiding Weak Keys in the BIKE Cryptosystem}

\author{Gretchen L.\ Matthews$^1$}
\address{$^1$Department of Mathematics, Virginia Tech, Blacksburg, VA 24060}
\email{gmatthews@vt.edu}

\author{Emily McMillon$^2$}
\address{$^2$Department of Mathematics, Rice University, Houston, TX 77005}
\email{em72@rice.edu}
\thanks{Matthews was partially supported by NSF DMS-2201075. McMillon was supported by NSF DMS-2303380. Both were partially supported by the Commonwealth Cyber Initiative.}



\date{May 9, 2025}

\dedicatory{}

\begin{document}

\begin{abstract}
Bit Flipping Key Encapsulation (BIKE) is a code-based cryptosystem  that was considered in Round 4 of the NIST Post-Quantum Cryptography Standardization process. It is based on quasi-cyclic moderate-density parity-check (QC-MDPC) codes paired with an iterative decoder. While (low-density) parity-check codes have been shown to perform well in practice, their capabilities are governed by the code’s graphical representation and the choice of decoder rather than the traditional code parameters, making it difficult to determine the decoder failure rate (DFR). Moreover, decoding failures have been demonstrated to lead to attacks that recover the BIKE private key. In this paper, we demonstrate a strong correlation between weak keys and $4$-cycles in their associated Tanner graphs. We give concrete ways to enumerate the number of $4$-cycles in a BIKE key and use these results to present a filtering algorithm that will filter BIKE keys with large numbers of $4$-cycles. These results also apply to more general parity check codes.
\end{abstract}

\maketitle

\section{Introduction} \label{sec:intro}

In recent years, there has been substantial work on the development of post-quantum cryptography to provide cryptosystems which are secure against attacks aided by a quantum computer. In the mid-1990s, it was established by Shor \cite{Shor_94, Shor_97} that quantum algorithms make systems like RSA and Diffie-Hellman key exchange insecure. He demonstrated a polynomial-time algorithm for solving the integer factorization problem and the discrete log problem upon which these cryptosystems are based. At that time, code-based cryptography had already been introduced with the advent of the McEliece cryptosystem \cite{McEliece}. The McEliece cryptosystem is defined using error-correcting codes\textemdash binary Goppa codes in particular. Its security is based on the NP-hard problem of decoding a general linear code \cite{BMvT_78} and the indistinguishability of binary Goppa codes, which are different from those assumptions supporting the security of RSA and elliptic curve cryptography.

While key sizes in the McEliece cryptosystem may be scaled to safeguard against an information set decoding attack aided by Grover's algorithm \cite{Grover_v_McEliece}, a quantum algorithm, the McEliece cryptosystem is believed to be post-quantum secure; indeed, the speed-up due to Grover's algorithm is quadratic. However, because of the large sizes of the public and private keys associated with desired security levels, there is a need to define code-based cryptosystems with smaller keys. 

The increased urgency to find quantum-safe public key cryptosystems has been driven by advances in quantum computing and incentivized by processes such as National Institute of Standards (NIST) Post-Quantum Cryptography Standardization. Among the Round 4 candidates in the NIST Post-Quantum Cryptography Standardization process was Bit Flipping Key Encapsulation (BIKE) \cite{BIKE22}, a code-based cryptosystem. It is based on quasi-cyclic moderate-density parity-check (QC-MDPC) codes, which were first considered for use in the McEliece cryptosystem in \cite{MTSB13}. This family of codes provides several advantages. The quasi-cyclic nature of the codes allows for their compact representation, enabling a matrix with $r^2$ (or $2r^2$) entries to be described by $d$ (or $2d$) $0-1$ bits, where $d < r$. By using a parity-check representation that is not too dense, iterative decoding, which has proven successful for low-density parity-check (LDPC) codes, may be used. In particular, a BIKE private key is a binary quasi-cyclic matrix $H \in \F_2^{r \times 2r}$, where $\F_2:=\{0,1\}$ denotes the binary field with two elements. The associated public key is its systematic form $H'$. Upon receipt of an encrypted message $H'\mathbf{e}^T$, where $\mathbf{e}$ is a vector of restricted weight associated with the message $\mathbf{m}$, the receiver uses the private key $H$ and a decoding algorithm to reveal $\mathbf{m}$. Hence, decoding is an integral part of BIKE. Unfortunately, as is the case with most iterative decoders, decoding failure can occur. 

Decoding failure was first harnessed to attack a QC-MDPC-based cryptosystem by  
Guo, Johansson, and Stankovski \cite{GJS16} in 2016 where they demonstrated that decoding failures can lead to a reaction attack. They note a relationship between the private key and decoding failure. They associate a distance spectrum with a private key, which tracks distances between nonzero entries of the key. Then they show that the private key can be determined from this distance spectrum. In particular, they note  a strong correlation between
the decoding error probability and vectors with certain distance profiles. We will build on this notion of distance spectrum to offer a perspective focused on combinatorial structures known to lead to iterative decoder failure. 

Since the work of \cite{GJS16}, there has been a focus on understanding weak keys. 
Sendrier and Vasseur
\cite{SV20} introduced three classes of weak keys, which have since been further analyzed in e.g. \cite{NSPZNGD23}. Additional classes of weak keys were introduced in \cite{WWW23, AYU20}.

Much work has focused on the decoding failure rate (DFR) of the BIKE system. In most cases, iterative decoder reliability and DFRs cannot be calculated explicitly.  One approach in the literature relies on extensive simulation results to try to extrapolate error-floor performance for the extremely small DFRs necessary for use in post-quantum cryptographic applications \cite{T18}. Other analyses include \cite{BBCPS21, SV20b}. Notably, in \cite{ABHLPR22}, the authors explore DFRs for relatively small versions of the BIKE cryptosystem in an attempt to extrapolate the behavior of the much larger versions of these codes that would actually be used in practice. The experimental data for their work is available publicly on Github, and  their data was used extensively to try to understand BIKE decoder failure in this present work. Alternative approaches include 
deriving a theoretical model to estimate the DFR, relying on some generically accepted heuristic assumptions, and determining a bound on the error correcting capability for each key, which does depend on additional assumptions; see \cite{Baldi_in_place, MACKAY_Postol_03, Santini, SV20b, T18, V21}. There are tradeoffs in these techniques for handling the DFR: the latter being the most reliable but with the most conservative parameters; the former providing some balance between good parameters and assumptions; and extrapolation giving the best parameters but subject to the validity of the assumptions on the DFR decay rate and behavior. For each approach, one may consider the role of weak keys. Loosely speaking, the idea is that because weak keys may exhibit a higher DFR than other keys, the average DFR may be decreased by filtering out weak keys, which could be used to prove that the DFR was negligible. At the same time, one must consider the impact on the overall key space and performance, in particular, to ensure that discarding weak keys does not facilitate an attack by key recovery, adjusting parameters to compensate for removed keys if necessary. 

In this paper, we provide a new filter for weak keys in the BIKE cryptosystem. The filter is based on combinatorial structures ($4$-cycles in an underlying graph) that may lead to decoder failure. We enumerate them directly from the polynomials that define the quasi-cyclic code upon which the private key is based. We provide experimental results demonstrating the utility of this approach. A primary contribution is the evidence that 4-cycles within a single circulant block of the parity-check matrix are more conducive to weak keys than 4-cycles across blocks. Empirical observations in some instances show that there are significantly more 4-cycles within each circulant block in the keys with decoding failures suggesting that focusing on each circulant block independently may be a better strategy to avoid weak keys.

This paper is organized as follows. Preliminaries are given in Section \ref{section:prelim}, including background on parity-check codes in Section~\ref{section:codes}, a review of BIKE in Section~\ref{section:BIKE}, and some results on the mathematical basis for weak keys in BIKE in Section~\ref{section:weak_keys}. Section~\ref{section:cycles_weak} contains results on counting $4$-cycles in the BIKE cryptosystem (Section~\ref{sec:bikecycles}) and more generally for any QC-MDPC code based cryptosystem (Section~\ref{sec:generalcycles}). Section \ref{section:filter} describes a new weak key filter for the structures identified in Section~\ref{sec:bikecycles}. In Section~\ref{section:ideal}, we explore the question of how and when $4$-cycles can be avoided completely in QC-MDPC codes. Section \ref{section:exp} contains experimental results and analysis. Finally, the paper concludes in Section~\ref{section:conclusions}.

\section{Preliminaries} \label{section:prelim}

In this section, we provide an overview of the BIKE cryptosystem focusing on facets which are relevant for this work. Before introducing the BIKE cryptosystem, we review relevant definitions, notation, and facts from coding theory.

Given a positive integer $n$, $[n]:=\left\{ 1, \dots, n \right\}$ and $[n]_0:=\left\{ 0, \dots, n-1 \right\}$. We use $\N$ to denote the set of nonnegative integers. The set of $m \times n$ matrices with entries in a ring $R$ is written as $R^{m \times n}$. The $i^{\text{th}}$ row and $j^{\text{th}}$ column of $A \in R^{m \times n}$ are denoted by $\row_iA$ and $\col_jA$; the transpose of $A$ is $A^T \in R^{n \times m}$. The $r \times r$ identity matrix is denoted $I_r$. Given a matrix $A$, the entry in the $i^{\text{th}}$ row and $j^{\text{th}}$ column is denoted $(A)_{i,j}$. Vectors are written in boldface, i.e. as $\mathbf{v}$. The support of a vector $\mathbf{v} \in \F_2^n$ is $\supp(\mathbf{v}):= \left| \left\{ i \in [n]: v_i \neq 0 \right\} \right|.$
Given an arbitrary $n$-tuple $\mathbf{x} = (x_1, \dotsc, x_n)$, $\mathbf{x}^i$ denotes the cyclic shift of $\mathbf{x}$ by $i$ positions to the right, i.e.,
\[ \mathbf{x}^i = (x_{n-i+1}, x_{n-i}, \dotsc, x_{n}, x_1, \dotsc, x_{n-i})\] and 
for an interval $[a,b) \subseteq [n],$
$\mathbf{x}^{[a,b]}:=(x_a, x_{a+1}, \dots, x_{b-1})$.
 A circulant matrix $A$ is a square matrix for which $\row_{i+1}A$ is a cyclic shift right to the right of $\row_iA$. 
We use the standard Hamming weight which for $\mathbf{v} \in \F_2^n$ is  $\text{wt}(\mathbf{v}):= \left| \supp (\mathbf{v}) \right|$. Two codes are permutation equivalent if and only if they are the same up to a fixed permutation of the codeword coordinates.

\subsection{Parity-check codes} \label{section:codes}

Given a matrix $H \in \F_2^{m \times n}$, the binary linear code with parity-check matrix $H$ is
$$C(H)=\left\{ \mathbf{c} \in \F_2^n \mid \mathbf{c} H^T = \mathbf{0} \right\}.$$
In other words, $C(H)$ is the nullspace of $H$. Iterative decoding algorithms operate on the Tanner graph of a code \cite{Tanner}, which depends on the parity-check matrix itself. The Tanner graph $T(H)$ of $H$ (or of $C(H)$) is the graph with biadjacency matrix $H$. In particular, $T(H)$ is the bipartite graph defined as $T(H) = (V \cup W; E)$, where the vertex set $V = \{v_1, \dotsc, v_n\}$ is the set of \textit{variable nodes} (also called \textit{bit nodes}, the vertex set $W = \{c_1, \dotsc, c_m\}$ is the set of $\textit{check nodes}$, and $E$ is the set of edges. The set $V$ of vertices corresponds to the columns of $H$, and the set $W$ of vertices corresponds to the rows of $H$. Vertices $v_i$ and $x_j$ are adjacent if and only if $(H)_{j,i}\neq 0$. In this setting, $\mathbf{c}=(c_1,c_2,\ldots,c_n)\in\F_2^n$ is a codeword of $C(H)$ if and only if the assignment of the values $c_1, c_2, \ldots, c_n$ to their corresponding variable nodes $v_1, v_2, \ldots, v_n$ of the Tanner graph satisfy
$$\sum_{i\in \mathcal{N}(v_j)}{ (H)_{j,i} c_i}= 0 \pmod 2$$
for all $j \in [m]$, where $\mathcal{N}(v_j)$ denotes the set of vertices adjacent to $v_j$ in $T(H)$. An example appears in Figure \ref{fig:Tanner_graph}.

\begin{figure}[h]
\begin{minipage}{0.45\textwidth}
\[ H = \begin{bmatrix} 1 & 1 & 0 & 1 & 0 & 0 \\ 0 & 1 & 1 & 0 & 1 & 0 \\ 0 & 0 & 0 & 1 & 1 & 1 \end{bmatrix} \]
\end{minipage} \begin{minipage}{0.45\textwidth} \centering

\begin{tikzpicture}[scale=.4,square/.style={regular polygon,regular polygon sides=4}]

\node (v1) at (0,0) [circle,draw,fill=black,scale=0.5] {};
\node (v2) at (2,0) [circle,draw,fill=black,scale=0.5] {};
\node (v3) at (4,0) [circle,draw,fill=black,scale=0.5] {};
\node (v4) at (6,0) [circle,draw,fill=black,scale=0.5] {};
\node (v5) at (8,0) [circle,draw,fill=black,scale=0.5] {};
\node (v6) at (10,0) [circle,draw,fill=black,scale=0.5] {};
\node (c1) at (2,3) [square,draw,fill=black,scale=0.5] {};
\node (c2) at (5,3) [square,draw,fill=black,scale=0.5] {};
\node (c3) at (8,3) [square,draw,fill=black,scale=0.5] {};

\path[b] (v1) to (c1);
\path[b] (v2) to (c1);
\path[b] (v2) to (c2);
\path[b] (v3) to (c2);
\path[b] (v4) to (c1);
\path[b] (v4) to (c3);
\path[b] (v5) to (c2);
\path[b] (v5) to (c3);
\path[b] (v6) to (c3);

\end{tikzpicture}
\end{minipage}
\caption{A parity check matrix $H$ (left) for a linear code and its corresponding Tanner graph $T(H)$ (right). Here and throughout, variable nodes are represented with circles and check nodes are represented with squares.}
\label{fig:Tanner_graph}
\end{figure}
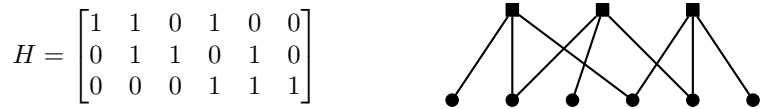

An important class of codes in application is the family of quasi-cyclic codes. In particular, the BIKE cryptosystem uses quasi-cyclic parity-check codes. 

\begin{definition}
    A \textit{quasi-cyclic code} is a linear code for which cyclically shifting a codeword a fixed number $n_0 \neq 0$ of  positions either to the right or to the left results is another codeword. For $n_0 = 1$, a quasi-cyclic code is a cyclic code. The integer $n_0$ is called the \textit{shifting constraint}. 
\end{definition}

There are several standard ways to represent a quasi-cyclic parity-check code. Given such a code with shifting constraint $n_0$ and positive integers $c$ and $t$ with $c < t$, let $H_i \in \F_2^{c \times n_0}$ for $i \in [t]$. Then the nullspace of  
\[ H = \begin{bmatrix} H_1 & H_2 & \dotsm & H_t \\
                       H_t & H_1 & \dotsm & H_{t-1} \\
                       \vdots & \vdots & \ddots & \vdots \\
                       H_2 & H_3 & \dotsm & H_1 \end{bmatrix} 
                       \in \F_q^{tc \times t n_0}
                       \]
is a quasi-cyclic parity-check code. Such a matrix is often represented in its circulant form, $H_c$, instead, where $C(H_c)$ is permutation equivalent to $C(H)$. The matrix $H_c$ is a $c$ by $n_0$ block matrix composed of $t \times t$ \textit{circulant} matrices
$J_{i,j} \in \F_2^{t \times t}$ for $i \in [c]$ and $j \in [n_0]$:
\begin{equation} H_c = \begin{bmatrix} J_{1,1} & J_{1,2} & \dotsm & J_{1,n_0} \\
                         J_{2,1} & J_{2,2} & \dotsm & J_{2,n_0} \\
                         \vdots & \vdots & \ddots & \vdots \\
                         J_{c,1} & J_{c,2} & \dotsm & J_{c,n_0} \end{bmatrix}.
\label{eq:circulant}
\end{equation}

The most commonly studied parity-check codes are \textit{low-density parity-check} (LDPC) codes, which are defined by sparse parity-check matrices, meaning matrices with relatively few nonzero entries, or alternatively. More precisely, an LDPC code may be defined by a parity check matrix with  row weights that are $O(1)$. When paired with iterative decoders, LDPC codes are capacity-approaching, meaning that they allow for communication near the Shannon limit so that the probability of lost information can be made as small as desired provided the noise is bounded \cite{Mackay1997near}. Moreover, LDPC codes can be decoded in time linear in their block length. Unfortunately, the speed and performance of LDPC code iterative decoders come at the cost of guaranteed decoder success. Even if there are relatively few errors in the received word, the decoder may fail to converge to a solution. Various Tanner graph substructures, such as pseudocodewords \cite{KS07}, trapping sets \cite{NCMV12}, stopping sets \cite{DPTRU02}, and absorbing sets \cite{DZAWN09}, have been shown to contribute to iterative decoder failure.

While suitable for other applications, when used in code-based cryptosystem, the sparisty of LDPC codes yield attacks \cite{Monico_LDPC_00, Baldi_ISIT_07, Baldi_08}. Consequently, proposed cryptosystems use a slightly denser analog of LDPC codes, \textit{moderate-density parity-check} (MDPC) codes. An MDPC code has a parity-check matrix representation $H$ with row weights $O(\sqrt{n})$, where $n$ is the length of the code. While MDPC codes avoid the attack deployed on LDPC-based systems, actual performance guarantees are elusive.

Iterative decoding is optimal on cycle-free graphs \cite{W96}, but unfortunately, codes with Tanner graph representations without cycles have poor distance parameters \cite{ETV99}. An endeavor when working with LDPC codes is to maximize the \textit{girth}, or length of the shortest cycle, of a code's Tanner graph. Graphs without short cycles are locally cycle-free, which leads to better decoder performance.

Much research exists on necessary and sufficient conditions for achieving quasi-cyclic LDPC (QC-LDPC) codes with given Tanner graph conditions, such as large girths and minimizing harmful graph substructures. In \cite{F04}, Fossorier gave the necessary and sufficient conditions for a QC-LDPC code's Tanner graph to have a given girth for the class of QC-LDPC codes whose circulants all have row weight $1$. Crucially, Fossorier showed that this class of parity-check codes has associated Tanner graphs with maximum girth $12$. Because Tanner graphs are bipartite, their minimum girth is $4$.

As parity-check matrices and their associated Tanner graphs become denser (i.e. closer to MDPC), it becomes more difficult to avoid short cycles. In fact, any quasi-cyclic parity-check code with a circulant matrix with row weight at least $3$ contains a $6$-cycle,  so the problem of ``optimizing'' such a code for girth is reduced to avoiding $4$-cycles.

It is generally noted that the performance of a parity-check code depends on not only the girth of the Tanner graph representing it but also on the number of short cycles (see, for instance, \cite{halford_06}). There are algorithms to count the number of short cycles (see \cite{Dehghan_20} and references therein) but to our knowledge do not provide a formula for the number of them. Our work makes extensive use of the quasi-cyclic structure of the parity-check codes to provide explicit ways to count short cycles. In Section~\ref{section:cycles_weak}, we will provide a framework to count $4$-cycles in the parity-check matrices used in BIKE.

\subsection{BIKE} \label{section:BIKE}

The BIKE cryptosystem \cite{BIKE22} was considered as a candidate in the National Institute of Standards (NIST) Post-Quantum Cryptography Standardization process, though not ultimately selected for standardization. The private key in BIKE is an $r \times 2r$ quasi-cyclic binary matrix $$
H=\left[ H_0 \ H_1\right] \in \F_2^{r \times 2r}
$$ composed of two circulant blocks, $H_0$ and $H_1$, each of size $r \times r$ with $r$ prime and so that $x^r - 1$ has only two irreducible factors modulo $2$. These circulant blocks have the form 
\begin{equation} H_i=\begin{bmatrix} 
    h_{i,0} & h_{i,1} & \dotsm & h_{i,r-1}\\
    h_{i,r-1} & h_{i,0} & \dotsm & h_{i,r-2}\\
    \vdots & \vdots & \ddots & \vdots \\
    h_{i,1} & h_{i,2} & \dotsm & h_{i,0}
    \end{bmatrix} \in \F_2^{r \times r}.
\label{eq:Hi_circulant}
\end{equation} 
Each column of $H$ (and hence of $H_i$, for $i \in \{0,1\}$) has weight $d$, and each row
of $H_i$, $i \in \{0,1\}$, has weight $d$
(so that each row of $H$ has weight $2d$). Because $H$ is quasi-cyclic, $\row_j H_i$ is a shift to the right of $\row_{j-1} H_i$ for all for $j \in [r-1]$, allowing $H_i$ to be represented by its first row,
$$\left( h_{i,0},  h_{i,1},  \dots,  h_{i,r-1} \right) \in \F_2^r,$$
or more succinctly by $\supp \left( \row_1(H_i)\right)$, the set of $d$ coordinates in the support of its first row:
$$\left\{ j \in [n] \mid h_{i,j} \neq 0 \right\}.$$ 
Furthermore, in order for $H$ to be considered the parity-check matrix of an MDPC code, the row weight $2d$ must be approximately the length of the code, i.e. $2d \approx \sqrt{2r}$. The public key associated with $H$ is its systematic form $H'= \begin{bmatrix} I_r & H_2\end{bmatrix}$ where
$$H_2=H_0^{-1}H_1 \in \F_2^{r \times r}.$$
Notice that because $H'$ is in systematic form, it can be completely described by the smaller matrix $H_2 \in \F_2^{r \times r}$. 

It is convenient to represent the circulant matrices $H_0$ and $H_1$ by polynomials 
\begin{equation} \label{eq:poly_description_of_rows}
  \begin{array}{ccc}
 \mathbf{h}_0:=  \left(h_{0,0}, h_{0,1}, \dots, h_{0,r-1} \right)&\leftrightarrow&
       h_0(x):=\sum_{j=0}^{r-1} h_{0,j}x^j \\ \ \\
  \mathbf{h}_1:=     \left(h_{1,0}, h_{1,1}, \dots, h_{1,r-1} \right)
       &\leftrightarrow&
       h_1(x):=\sum_{j=0}^{r-1} h_{1,j}x^j
\end{array}
\end{equation}
 where
$$h_0(x), h_1(x) \in {\F_2[x]}/{\left< x^r -1 \right>},$$ 
noting that both $h_0(x)$ and $h_1(x)$ have few (in fact $d$) nonzero terms. For $i \in \{0,1\}$, we say that $h_i$ has support $\supp (h_i):=\left\{ j \mid h_{i,j} \neq 0 \right\}$ and weight $\left| h_i \right| := \left| \supp (h_i) \right|$. The parameters suggested for BIKE are given in Table \ref{table:BIKE_params}. 

\begin{table*}[t]
\begin{center}
\begin{tabular}{|l|c|c|c|}
% \toprule
\hline
parameter  & $\lambda=128$ & $\lambda=192$ & $\lambda=256$       \\  \hline
 $r$,   size of circulant  & $12,323$ & $24,659$ & $40,973$       \\  \hline
$d$, row weight of parity-check matrix  & 142 & 206 & 274 \\ \hline
$t$, error vector weight   &  134 & 99 & 264  \\  \hline 
% \bottomrule
\end{tabular}\\ \ \\
\caption{Suggested BIKE parameters for $\lambda$-bit security.}
 \label{table:BIKE_params}
% \end{tiny}
\end{center}
\end{table*}

Encryption of a message $\mathbf{m} \in \F_2^k$ begins with associating an error vector $\mathbf{m} \in \F_2^t$ of predetermined weight $t$, with $\mathbf{m}$ and then determining the associated syndrome $$H'\mathbf{e}^T$$ as in the approach by Niederreiter \cite{NIEDERREITER_86}. Alternatively, one may compute a generator matrix $G$ of the code associated with $H$ as in \cite{McEliece, NSPZNGD23} and encrypt $\mathbf{m}$ as $\mathbf{m}G+\mathbf{e}$. In either case, the receiver must perform a decoding algorithm to uncover the message $\mathbf{m}$. BIKE uses an iterative decoder, motivated by prior results demonstrating superior performance for LDPC codes when coupled with iterative decoders as explained in Section \ref{section:codes}. Unfortunately, at times, decoding failures occur due to the ways in which local information is utilized to ultimately produce a global decision. 

The keys and algorithms described in this paper focus on the Black-Gray Flip (BGF) Decoder \cite{DGK20}. In a standard bit-flipping decoder, on any given iteration, any bit node with a large enough number of unsatisfied check node neighbors will be flipped, meaning a $0$ is changed to a $1$ or vice versa. The BGF decoder begins with a typical round of bit flipping in which the decoder keeps track of which bits were flipped (black bits) and which bits were ``close'' to being flipped (gray bits). It then runs two additional rounds of bit flipping, checking only the black bits, then the gray bits, to see if they now have large numbers of unsatisfied check node neighbors. Then the BGF decoder runs a small number of standard bit-flipping decoder steps. 

In the next subsection, we will review the notions of weak keys in the literature in preparation for comparing them with the combinatorial structures leading to decoding weaknesses that we identify in Section \ref{section:cycles_weak}. 

\subsection{Weak Keys in BIKE}
\label{section:weak_keys}

The notion of weak keys was introduced in \cite{DGK20} as a way of capturing some decoding failures, particularly those resulting from keys with identifiable structural patterns. In \cite{BIKE22}, they are differentiated from decoding failures resulting from structured errors. They are further studied in e.g. \cite{SV20, NSPZNGD23, WWW23, AYU20}.

The keys in BIKE can be thought of as corresponding to the first row of the parity-check matrix $H$, as described in Expression (\ref{eq:poly_description_of_rows}). The \textit{weak keys} for the BIKE cryptosystem identified in \cite{DGK20,V21} are described as:
\begin{itemize}
    \item Type I: keys with many consecutive nonzero bits,
    \item Type II: keys with nonzero bits at regular intervals,
    \item Type III: keys with many column intersections between the two cyclic blocks.
\end{itemize}
In \cite{ABHLPR22}, the authors examine the iterative decoding performance of QC-MDPC codes for use in the BIKE cryptosystem, focusing on experimental findings on the decoding failure rate. One may note, as in \cite{SV20}, that Type I weak keys are also Type II weak keys. However, they are distinguished in the literature due to observed differences in their impact on the DFR.

There are other notions of weak key associated with BIKE, which we mention here to avoid confusion with those we consider. In \cite{AYU20}, the authors note that when using a scheme in which
$$ H= \begin{bmatrix} H_0 & H_1 \end{bmatrix} $$
is the private key associated with public key 
$$ H'=\begin{bmatrix} I \ H_2 \end{bmatrix}, $$
one may use factorizations of the polynomials $h(x)$ defining the circulant matrix $H_2$ to determine the polynomials associated with $H_0$ and $H_1$. These polynomials are factors of $x^r-1$. However, as pointed out above, $r$ is selected as a large prime such that $x^r-1$ has only two factors over $\F_2$. Hence, the term weak key used here differs from that in \cite{AYU20}. In \cite{Bardet_weak_16}, the authors apply the extended Euclidean algorithm to uncover a private key from a public key. They use the term weak key to mean any pair of private keys that can be recovered using the extended Euclidean algorithm from public data, ultimately taking weak key to mean any $(h, h') \in \left({\F_2[x]}/{\left< x^r -1 \right>}\right)^2$, each having $w$ nonzero terms, with $\deg h+ \deg h' < r$. As stated in \cite{NSPZNGD23}, these types of weak keys are irrelevant to the IND-CCA security. Hence, we do not consider them here. 

In \cite{WWW23}, the authors introduce the concept of the gathering property to define a set of weak keys. For $m \in [r]$ and $\epsilon \in \N$, set
\[K_{m,\epsilon}:= \left\lbrace (\mathbf{h}_0, \mathbf{h}_1) \in \left( {\F_2[x]}/{\left< x^r -1 \right>}
\right)^2 \mid \exists a \in \N \text{ with } \text{wt}(\mathbf{h}_0^{[a,a+m)}) = \text{wt}(\textbf{h}_0) - \epsilon \right\rbrace\]
where the interval $[a,a+m)$ is taken modulo $r$. Loosely speaking, these keys are such that all or most of the nonzero entries of $\mathbf{h}_0$ are contained within a ``relatively'' small portion of the interval from $0$ to $r-1$. In \cite{WWW23}, $\epsilon$ is taken to be $0$ or $1$.\\

\begin{definition} \label{def:gatheringproperty}
A key $(\mathbf{h}_0,\mathbf{h}_1) \in \left({\F_2[x]}/{\left< x^r -1 \right>}\right)^2$ is said to satisfy the $(m,\epsilon)$-gathering property if and only if $(\mathbf{h}_0,\mathbf{h}_1) \in K_{m,\epsilon}$. 
\end{definition}

The notion of a distance spectrum of a key in BIKE was introduced in \cite{GJS16} and further refined in \cite{SV20}. We provide following definition of it and related terms. \\

\begin{definition} \label{def:distance_mult_profile}
    Let $H$ be a circulant matrix and let $h \in \F_2[x]/\langle x^r - 1 \rangle$ be the polynomial associated with its first row.
    \begin{enumerate}
        \item The distance between two positions $i,j \in [r]_0$ in H is given by
        \[ d(i,j) = \min\{ \pm \left( j-i \right)\mod r\}.\]
        \item The multiplicity of a distance is the number of times it appears as the difference of two degrees of nonzero monomials of $h$,
        \[\mu(\delta,h) = \left| \left\lbrace (i,j) \mid d(i,j) = \delta, h_i = h_j = 1 , 0 \leq i \leq j < r \right\rbrace \right|.\]
        \item The spectrum of $h$ is the set of all nonzero distances with their multiplicity,
        \[ \spec(h) = \left\lbrace (\delta, \mu(\delta,h) \mid \delta \in \left\lbrace 1, \dotsc, \left\lfloor r/2 \right\rfloor \right\rbrace \right\rbrace.\]
    \end{enumerate}
\end{definition}

We introduce some additional terms that expand on those in Definition~\ref{def:distance_mult_profile} and will allow us to provide additional insight into weak keys.\\

\begin{definition} \label{def:distance_mult_profile2}
    Let $H$ be a circulant matrix and let $h \in \F_2[x]/\langle x^r -1 \rangle$ be the polynomial associated with its first row.
    \begin{enumerate}
        \item The full spectrum of $h$ is 
        \[ \overline{\spec}(h) = \left[ \delta \mid d(i,j) = \delta, h_i = h_j = 1, 0 \leq i < j < r \right],\]
        the multiset of all nonzero distances between nonzero monomials of $h$,
        \item The full multiplicity of a distance $i$ is the number of times a distance appears in $\overline{\spec}(h)$, or, equivalently,
        \[ \gamma(i,h) = \left| \left\lbrace (\delta, \mu(\delta,h)) \mid \mu(\delta,h) = i\right\rbrace \right|.\]
        In other words, this is the multiplicity of a multiplicity of a distance.
        \item The distance multiplicity spectrum of $h$ with $\text{wt}(h) = d$ is
        \[ \text{mSpec}(h) = \left\lbrace (i, \gamma(i,h)) \mid i \in [d]_0\right\rbrace.\]
    \end{enumerate}
\end{definition}

Notice that in Definition~\ref{def:distance_mult_profile2}, $\gamma(i,h) \leq d-1$, because the number of times a distance can appear is upper bounded by $d-1$. Furthermore, $\gamma(i,h) = d-1$ exactly when there exists $\ell \in \N$ with $d(i,j) = \ell$ for all sequential nonzero indices of $h$.

We illustrate Definitions~\ref{def:distance_mult_profile} and \ref{def:distance_mult_profile2} in the next example.\\

\begin{example}
    Let $h = 1 + x + x^8 + x^9 \in \F_2[x]/\langle x^{11} - 1 \rangle$. The associated vector representation of $h$ is $\mathbf{h} = (1, 1, 0, 0, 0, 0, 0, 0, 1, 1, 0)$, and the corresponding circulant parity-check matrix $H$ is given by
    \[ H = \begin{bmatrix} 1 & 1 & 0 & 0 & 0 & 0 & 0 & 0 & 1 & 1 & 0 \\
                           0 & 1 & 1 & 0 & 0 & 0 & 0 & 0 & 0 & 1 & 1 \\
                           1 & 0 & 1 & 1 & 0 & 0 & 0 & 0 & 0 & 0 & 1 \\
                           1 & 1 & 0 & 1 & 1 & 0 & 0 & 0 & 0 & 0 & 0 \\
                           0 & 1 & 1 & 0 & 1 & 1 & 0 & 0 & 0 & 0 & 0 \\
                           0 & 0 & 1 & 1 & 0 & 1 & 1 & 0 & 0 & 0 & 0 \\
                           0 & 0 & 0 & 1 & 1 & 0 & 1 & 1 & 0 & 0 & 0 \\
                           0 & 0 & 0 & 0 & 1 & 1 & 0 & 1 & 1 & 0 & 0 \\
                           0 & 0 & 0 & 0 & 0 & 1 & 1 & 0 & 1 & 1 & 0 \\
                           0 & 0 & 0 & 0 & 0 & 0 & 1 & 1 & 0 & 1 & 1 \\
                           1 & 0 & 0 & 0 & 0 & 0 & 0 & 1 & 1 & 0 & 1 
    \end{bmatrix}.\]
    We compute $d(i,j)$ for all pairs $i,j \in \{0, 1, 8, 9\}$:
    \[ d(0,1) = 1, \,\, d(0,8) = 3, \,\, d(0,9) = 2, \,\, d(1,8) = 4, \,\, d(1,9) = 3, \,\, d(8,9) = 1\]
    and use these to determine the multiplicities of each distance $\delta \in \{1, 2, 3, 4, 5\}$:
    \[\mu(1,h) = 2, \quad \mu(2,h) = 1, \quad \mu(3,h) = 2, \quad \mu(4,h) = 1, \quad \mu(5,h) = 0.\]
    Hence, 
    \[ \spec(h) = \{(1,2), (2,1), (3,2), (4,1), (5,0)\}.\]
    The full spectrum of $h$ is the multiset of nonzero distances
    \[ \overline{\spec}(h) = \left[ 1, 1, 2, 3, 3, 4\right].\]
    The distance multiplicities of each distance are 
    \[ \gamma(1,h) = 2, \quad \gamma(2,h) = 2, \quad \gamma(3,h) = 0.\]
    And finally, the distance multiplicity spectrum of $h$ is
    \[ \text{mSpec}(h) = \{(1,2), (2,2), (3,0)\}.\]
\end{example}

According to \cite{SV20}, Type I weak keys have several multiplicities of high frequency. While there are fewer Type I weak keys than Type II weak keys, they have a more adverse impact on DFR than Type II weak keys with a high frequency multiplicity for a particular distance.

In the next section, we will develop the concepts needed to provide a new perspective on the notion of weak keys.

\section{Cycles and weak keys}
\label{section:cycles_weak}

In this section, we will use combinatorial methods to provide an analysis of weak keys, with a particular connection to $4$-cycles in the associated code's Tanner graph.  In Section~\ref{sec:bikecycles}, we look specifically at the cycle structure of weak keys in the BIKE cryptosystem. In Section~\ref{sec:generalcycles}, we extend these results to QC-MDPC codes more generally.

We begin with a definition to be used in this section's results.\\

\begin{definition} \label{def:col_intersection}
    Let $H, H' \in \F_2^{r \times r}$.
    \begin{enumerate}
        \item The column intersection of columns $i$ and $j$ is 
        \[ \text{CI}(i,j) = \left| \supp \col_i H \cap \supp \col_j H \right|,\]
the number of indices in which column $i$ and column $j$ are both equal to $1$. 
        \item The maximum column intersection between $H$ and $H'$ is 
        \[ \text{MCI}(H,H') = \max_{i,j \in [r]_0} \left| \supp \col_i H \cap \supp \col_j H' \right|,\]
        the maximum intersection between any column in $H$ with any column in $H'$.
    \end{enumerate}
\end{definition}

The correspondence between a polynomial $h \in \F_2[x] / \langle x^r - 1 \rangle$ and its corresponding circulant parity-check matrix $H \in \F_2^{r \times r}$ is defined via row shifts\textemdash the first row corresponds to indices of the nonzero coefficients of $h$, and each row thereafter in $H$ is a shift of the row above by one index to the right. Most of the results in this section focus on columns and column intersections, not rows and row intersections. For this reason, we begin with Remark~\ref{rem:rowcoleq}, which explains the correspondence between row and columns of a circulant matrix $H$.\\

\begin{remark} \label{rem:rowcoleq}
    Given some $h \in \F_2[x]/\langle x^r - 1 \rangle$ and associated circulant parity-check matrix $H$, there is a correspondence between the columns of $H$ and the rows of $H$. In particular, the row $\row_iH$ is the same as column $\mathbf{c}_{r - 1 - i}^{-1}$, where the subscript is taken modulo $r$ and the ${}^{-1}$ corresponds to inverting the order of the entries of the vector. To see this, notice that the circulant of each $1$ is at a diagonal, and so the matrix $H$ is symmetric about its opposite diagonal. Hence, there is a distance-preserving correspondence between the columns and rows of the circulant parity-check matrix $H$. 
\end{remark}

\subsection{BIKE} \label{sec:bikecycles}

We begin by characterizing the column intersections of a single circulant matrix.\\

\begin{proposition} \label{prop:maxcolint}
    Let $h \in \F_2[x]/ \langle x^r - 1 \rangle$ and $H \in \F_2^{r \times r}$ be the circulant parity-check matrix corresponding to $h$. Then, for $i,j \in [r]_0$,
    \[ CI(i,j) = \mu(d(i,j),h).\]
    In particular,
    \[ MCI(H) = \max_{i,j \in [r]_0} \{\mu(d(i,j),h)\} = \max_{j  \in [r]_0} \{\mu(d(0,j),h)\}.\]
\end{proposition}

\begin{proof}
Consider $\row_iH$ and $\row_jH$ with $i < j$. Let $\ell$ be any index such that $(H)_{i,\ell} = (H)_{j,\ell} = 1$ (i.e. any row intersection). Because $j > i$, $\row_jH$ is a cyclic shift of $\row_iH$ by $(j-i) \mod r$ positions, and so $(H)_{i,\ell - (j-i)}=1$ as well. But $\row_iH$ is a cyclic shift of row $\row_0H$ by $i$ positions, and so  $(H)_{0,\ell - j} =  (H)_{0,\ell - i} = 1$ as well. Hence positions $\ell - j < \ell - i$ are nonzero in $\row_0H$. Note that
    \begin{align*}
        d(\ell-j,\ell-i) &= \min\{ (\ell - i) - (\ell - j) \mod r, (\ell - j) - (\ell - j) \mod r \} \\
        &= \min\{j - i \mod r, i - j \mod r\} \\
        &= d(i,j).
    \end{align*}
    Hence, each pair contributing to the row intersection of  $\row_iH$ and $\row_jH$ comes from some pair of nonzero coordinates of $\row_0H$ with distance $d(i,j)$. This argument also works in the backward direction and shows that the row intersection of rows $\row_iH$ and $\row_jH$ is $\mu(d(i,j),h)$. By Remark~\ref{rem:rowcoleq}, an equivalent statement holds for columns, and so $CI(i,j) = \mu(d(i,j),h)$. 
Then
    $MCI(H) = \max_{i,j \in [r]_0} \{ \mu(d(i,j),h) \}$  holds by definition. Now note that $\col_{i+k}H$ is a shift by $k$ (mod $r$) of $\col_iH$. In particular, comparing $\col_0H$ to some $\col_kH$ (i.e. computing the max column intersection of $\col_0H$) is equivalent to comparing $\col_iH$ to $\col_{i+k}H$. In other words, it suffices to only calculate column intersections between $\col_0H$ and columns $\col_iH$ with $i \in \left[ \lfloor r/2 \rfloor \right]$, so $MCI(H) = \max_{j \in [r]_0} \{ \mu(d(0,j),h)\}$.
\end{proof}

Proposition~\ref{prop:maxcolint} demonstrates that all columns of a circulant matrix $H$ have the same maximum column intersection between the other columns in $H$. Therefore, finding the maximum column intersection of one column is sufficient to find the column intersection numbers of all other columns of a circulant matrix $H$.

In the next result, we provide a formula for counting $4$-cycles within a single circulant matrix. Notice that given a parity-check matrix $H \in \F_2^{r \times r}$, a $4$-cycle in the Tanner graph $T(H)$ arises precisely when there exist $i, j \in [r]$, $i \neq j$, with $\mid \supp \left( \row_iH \right) \cap  \supp \left( \row_jH \right) \mid \geq 2$.\\

\begin{lemma} \label{lem:4cycles}
    Let $h \in \F_2[x]/\langle x^r - 1 \rangle$ and $H \in \F_2^{r \times r}$ the circulant parity-check matrix corresponding to $h$. The number of $4$-cycles in the Tanner graph of $H$ is
    \begin{equation} \label{eq:4cycles}
        r \cdot \sum_{\delta=1}^{\left\lfloor r/2 \right\rfloor} \binom{\mu(\delta,h)}{2}.
    \end{equation}
\end{lemma}

\begin{proof}
    Let $\row_iH = (w_{i,1}, \dotsc, w_{i,r})$. Notice that given row $\row_nH$, all subscripts of $\row_mH$ have been shifted from $\row_nH$ by $m-n \mod r$ positions, i.e. $w_{m,i} = w_{n, i+(m-n)}$, all subscripts taken $\mod r$.
    
    We claim that it suffices to count the $4$-cycles with some entries in row $0$ to find all $4$-cycles in $H$. Let $j$ and $k$ be such that $j \neq k$ and $w_{0,j} = w_{0,k} = 1$ and $\ell \in [r-1]$ be such that $w_{\ell, j} = w_{\ell, k} = 1$ (i.e. the entries $w_{0,j}, w_{0,k}, w_{\ell, j}$, and $w_{\ell, k}$ give rise to the edges of a $4$-cycle in the Tanner graph of $H$). Then, for any $m$, the entries $w_{m,j+m}, w_{m,k+m}, w_{\ell+m,j+m}$, and $w_{\ell + m, k + m}$ also form a $4$-cycle. Because $m \in \{0, \dotsc, r-1\}$ are all then contributing to a $4$-cycle, for each $4$-cycle with entries in $\col_0H$, there are $r$ such $4$-cycles in $H$. In the other direction, if $w_{n,j}, w_{n,k}, w_{m,j}$, and $w_{m,k}$ form the edges of a $4$-cycle, then so do $w_{0,j-n}, w_{0,k-n}, w_{m-n,j-n}$, and $w_{m-n,k-n}$, and hence they were counted when counting those from the first row. 

    Now we count the $4$-cycles with entries in $\row_0H$. By Proposition~\ref{prop:maxcolint}, the row intersection between $\row_0H$ and $\row_iH$ is $\mu(d(0,i),h)$, and so there are $\mu(d(0,i),h)$ pairs of ones in common between $\row_0H$ and $\row_iH$. Selecting any two of these gives a unique $4$-cycle between the two rows. So the number of $4$-cycles between $\row_0H$ and $\row_iH$ is $\binom{\mu(d(0,i),h)}{2} = \binom{\mu(i,h)}{2}$. We will now show that the $4$-cycles between $\row_0H$ and $\mathbf{w}_{r-i}$ are the same as those between $\row_0H$ and $\row_iH$, and so they should only be counted once. To see this, suppose $w_{0,k}$, $w_{0,\ell}$, $w_{r-i,k}$, and $w_{r-i,\ell}$ correspond to vertices of a $4$-cycle. Then taking cyclic shifts, $w_{i,k+i}$, $w_{i,\ell + i}$, $w_{0,k+i}$, and $w_{0,\ell + i}$ also correspond to a $4$-cycle between rows $\row_0H$ and $\row_iH$. 

    Hence, the total number of $4$-cycles with entries in $\row_0H$ is $\sum_{\delta=1}^{\lfloor r/2 \rfloor} \binom{\mu(\delta,h)}{2}$. By the earlier argument in this proof, this number multiplied by $r$ gives all $4$-cycles in the Tanner graph. The result follows.
\end{proof}

We next give a reformulation of the equation given in Lemma~\ref{lem:4cycles} in terms of the full multiplicities (see Definition~\ref{def:distance_mult_profile2}) which give the number of times a given distance multiplicity appears in $\spec(h)$. \\

\begin{corollary} \label{rem:4cycleequiv}
Given $h \in \F_2[x]/\langle x^r - 1 \rangle$ and $H \in \F_2^{r \times r}$ the circulant parity-check matrix corresponding to $h$, the number of $4$-cycles in the Tanner graph of $H$ is
    \[ r \cdot \sum_{i=2}^{d-1} \gamma(i,h) \binom{i}{2}.\]
\end{corollary} 
    
\begin{proof}
    Notice that if some multiplicity $\mu \in [d]_0$ appears $i$ times in (\ref{eq:4cycles}), in the sum, this multiplicity contributes $i \binom{\mu}{2}$ $4$-cycles. Further notice that $\mu$ cannot be more than $d-1$, as no two rows in $H$ can be identical. Also, because $\binom{0}{2} = \binom{1}{2} = 0$, we can begin the sum at $i=2$.
\end{proof}

Proposition~\ref{prop:maxcolint} and Lemma~\ref{lem:4cycles} consider the column intersections and $4$-cycles within a single circulant matrix in a BIKE private key. The next step is to consider the column intersections and $4$-cycles between the two circulant matrices that make up a private key. Proposition~\ref{prop:maxcolint2} gives bounds on the maximum column intersection based on the full spectrum of $h_0$ and $h_1$. Lemma~\ref{lem:4cycles2} immediately follows and provides an exact value for the maximum column intersection as well as way to count the number of $4$-cycles between the two circulants.\\

\begin{proposition} \label{prop:maxcolint2}
    Let $h_0, h_1 \in \F_2[x]/\langle x^r-1 \rangle$ and let $H_0, H_1 \in \F_2^{r \times r}$ be the circulant parity-check matrices corresponding to $h_0, h_1$, respectively. The maximum column intersection between  $H_0$ and  $H_1$ satisfies the following:
    \begin{enumerate}
        \item if $\overline{\spec}(h_0) \cap \overline{\spec}(h_1) \neq \varnothing$, $$2 \leq MCI(H_0, H_1)  \leq 2 \cdot \left| \overline{\spec}(h_0) \cap \overline{\spec}(h_1) \right|.$$
        \item $MCI(H_0, H_1) = 1$ if and only if $\overline{\spec}(h_0) \cap \overline{\spec}(h_1) = \varnothing$.
    \end{enumerate}
\end{proposition}

\begin{proof}
    Suppose  $\delta \in \overline{\spec}(h_0) \cap \overline{\spec}(h_1)$. Then there are pairs $h_{0,i}$, $h_{0,j}$ and 
    $h_{1,i'}$, $h_{1,j'}$
    with $d(h_{0,i},h_{0,j}) = d(h_{1,i'},h_{1,j'}) = \delta$. Taking a cyclic shift $h_0'$ of $h_0$ by $i' - i \mod r$ gives nonzero entries $h_{0,i+(i'-i)}' = h_{0,i'}'$ and $h_{0,j + (i'-i)}' = h_{0,j'}$. So $h_0'$ and $h_1$ have a column intersection of at least two.

    By the same argument, for any $\delta \in \overline{\spec}(h_0) \cap \overline{\spec}(h_1)$, an intersection of size two between a shift of $h_0$ and $h_1$ (or the reverse) can result in two indices of intersection between the columns. In the most extreme case, each pair of distance $\delta$ could result in two indices of intersection, proving 1.

    If $\overline{\spec}(h_0) \cap \overline{\spec}(h_1) = \varnothing$, then  $MCI(H_0,H_1) \leq 1$. As long as $h_0$ and $h_1$ have at least one nonzero index, the intersection is certainly nonzero, proving 2.
\end{proof}

\begin{lemma} \label{lem:4cycles2}
    Let $h_0, h_1 \in \F_2[x]/\langle x^r - 1\rangle$ and let $H_0, H_1 \in \F_2^{r \times r}$ be the circulant parity-check matrices corresponding to $h_0$, $h_1$, respectively. Define 
    \[ c_i := \left| \supp (h_0) \cap \supp (h_1^{i}) \right| \quad \text{for } i \in [r]_0.\]
    Then
    \begin{enumerate}
        \item $MCI(H_0, H_1) = \max_i c_i$ and
        \item the number of $4$-cycles between $H_0$ and $H_1$ is
        \[ r \cdot \sum_{i=0}^{r-1} \binom{c_i}{2}.\]
    \end{enumerate}
\end{lemma}

\begin{proof}
    By Remark~\ref{rem:rowcoleq}, we can treat row vectors as column vectors, and so can associate $h_0$ with the $(r-1)$st column of $H_0$ and $h_1^0$ with the $(r-1)$st column of $H_1$. Because $h_1^i$ for $i \in [r]_0$ correspond to the $r$ columns of $H_1$, the set of $c_i$'s correspond to the intersections of the supports of the final column of $H_0$ and each of the columns of $H_1$. We claim that this is the set of all column intersections between $H_0$ and $H_1$. Let $\mathcal{c}_{i,j}$ be the $j$th column of $H_i$.  Similarly to in the proof of Proposition~\ref{prop:maxcolint}, note that column $\mathbf{c}_{i,j+k}$ is a cyclic shift by $k$ (mod $r$) of column $\mathbf{c}_{i,j}$. In particular, comparing some arbitrary $\mathbf{c}_{0,j}$ with some arbitrary $\mathbf{c}_{1,j'}$ is equivalent to comparing $\mathbf{c}_{0,r-1}$ with $\mathbf{c}_{1,j'+r-1-j}$, subscripts taken modulo $r$. So all column intersections are calculated in this manner, and hence, the maximum of all these intersections gives the maximum column intersection between $H_0$ and $H_1$, proving (1).

    Selecting any two indices in $\supp(h_0) \cap \supp(h_1^i)$ gives rise to a unique $4$-cycle between the final column in $H_0$ and the $i+1 \mod r$ column in $H_1$. Hence, between these two columns, we have $\binom{c_i}{2}$ $4$-cycles. Each of these $4$-cycles is clearly unique, as all involve a fixed column in $H_0$ and a different column in $H_1$, so there are $\sum_{i=0}^{r-1} \binom{c_i}{2}$ $4$-cycles between the final column of $H_0$ and all columns of $H_1$. There are $r$ columns of $H_0$, each a cyclic column shift of the final column of $H_0$, and each contributing $\sum_{i=0}^{r-1} \binom{c_i}{2}$ $4$-cycles between the set column of $H_0$ and all columns of $H_1$. This gives that between $H_0$ and $H_1$ there are in total $r \cdot \sum_{i=0}^{r-1} \binom{c_i}{2}$ $4$-cycles.
\end{proof}

\begin{remark}
    The reader may notice that in Lemma~\ref{lem:4cycles2}, the indices range over values in $[r]_0$, while in Lemma~\ref{lem:4cycles}, the indices range over values $\left\{ 0, \dots, \left\lfloor \frac{r}{2} \right\rfloor \right\}$. Intuitively speaking, this is because for any $4$-cycle with entries in the first column/row of a single circulant, any set of $r$ $4$-cycles coming from a set of shifts of this $4$-cycle has two shifts which result in having entries in the first row/column of that circulant.
\end{remark}

We are now ready to state the main theorem of this section, which gives a concise way to count the number of $4$-cycles in a BIKE private key matrix $H = \left[ H_0 \ H_1 \right]$.\\

\begin{theorem} \label{co:all4cycles}
    Let $h_0, h_1 \in \F_2[x]/\langle x^r - 1\rangle$ and let $H_0, H_1 \in \F_2^{r \times r}$ be the circulant parity-check matrices corresponding to $h_0$, $h_1$, respectively. Define 
    \[ c_i := \left| \supp (h_0) \cap \supp (h_1^{i}) \right| \quad \text{for } i \in [r]_0.\]
    Then the total number of $4$-cycles in the matrix $H = \begin{bmatrix} H_0 & H_1 \end{bmatrix}$ is 
    \[ r \cdot \left( \sum_{\delta=1}^{\left\lfloor r/2 \right\rfloor} \binom{\mu(\delta,h_0)}{2} + \sum_{\delta=1}^{\left\lfloor r/2 \right\rfloor} \binom{\mu(\delta,h_1)}{2} + \sum_{i=0}^{r-1} \binom{c_i}{2} \right). \]
\end{theorem}

\begin{proof}
    This is immediate from Lemma~\ref{lem:4cycles} and Lemma~\ref{lem:4cycles2}.
\end{proof}

\subsection{General QC-MDPC Codes} \label{sec:generalcycles}

The BIKE cryptosystem is easily generalized to use a private key with an arbitrary number of circulant submatrices. We can establish results analogous to Lemma~\ref{lem:4cycles2} and Theorem~\ref{co:all4cycles} for general quasi-cyclic MDPC codes.

Let $h_{i,j}(x) \in \F_2[x]/\langle x^r-1 \rangle$ for $i \in \{0, \dotsc, c-1\}$ and $j \in \{0, \dotsc, n_0-1\}$ and let $H_{i,j} \in \F_1^{r \times r}$ be the corresponding circulant matrices. For full generality, we also assume that the $h_{i,j}$ do not have the same weights, i.e., it may happen that $|h_{i,j}| := d_{i,j} \neq d_{i',j'}=: |h_{i',j'}|$. Let
\begin{equation} H = \begin{bmatrix} H_{0,0} & H_{0,1} & \dotsm & H_{0,n_0-1} \\
                       H_{1,0} & H_{1,1} & \dotsm & H_{1,n_0-1} \\
                       \vdots & \vdots & \ddots & \vdots \\
                       H_{c-1,0} & H_{c-1,1} & \dotsm & H_{c-1,n_0-1}\end{bmatrix} \label{eq:qcbigmatrix}  \in \F_2^{rc \times rn_0}\end{equation}
be the quasi-cyclic matrix built from the $H_{i,j}$. We sometimes refer to the $H_{i,j}$ as blocks and their rows (resp., columns) as block rows (resp., block columns). We will generalize the results from the previous section to count the number of total $4$-cycles in $H$.

There are more ways for $H$ to have $4$-cycles when $n_0>2$ than there were in the case where $H$ consisted only of two circulant submatrices. We illustrate this in Figure~\ref{fig:matrix4cycles}. In particular, there are the following types of $4$-cycles:
\begin{itemize}
    \item Type A  are contained entirely within a single circulant.\\
    \item Type B  have two entries in the same column in two circulant subblocks which share a block row.\\
    \item Type C  have two entries in the same row in two circulant subblocks which share a block column. \\
    \item Type D  have one entry in each of four different circulant submatrices, where these submatrices pairwise share block row/columns.
\end{itemize}
 In the case of BIKE, only type A and type B $4$-cycles are present. With additional block rows in $H$ with $n_0>2$, there are two more types of $4$-cycles that can appear.

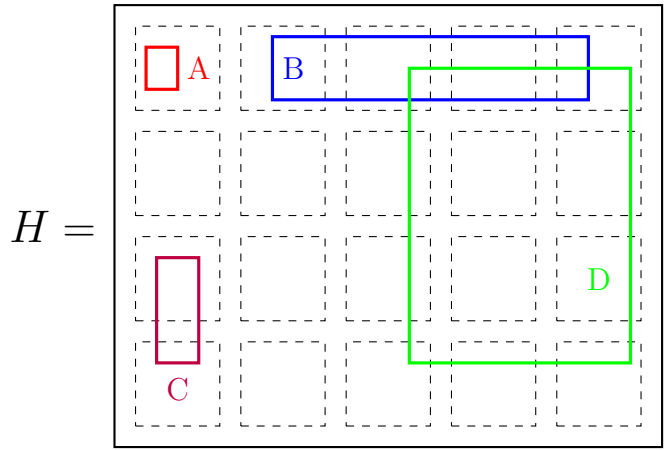
\begin{figure}[ht] \centering \label{fig:matrix4cycles}
    \begin{tikzpicture}[scale=1.4]
        \node[label=left:{\huge{$H=$}}] at (-0.1,2) {};
    
        \draw[black,thick] (-0.1,-0.1) rectangle (5.1,4.1);
        \draw[black,thin,dashed] (0.1,0.1) rectangle (0.9,0.9);
        \draw[black,thin,dashed] (0.1,1.1) rectangle (0.9,1.9);
        \draw[black,thin,dashed] (0.1,2.1) rectangle (0.9,2.9);
        \draw[black,thin,dashed] (0.1,3.1) rectangle (0.9,3.9);
        \draw[black,thin,dashed] (1.1,0.1) rectangle (1.9,0.9);
        \draw[black,thin,dashed] (1.1,1.1) rectangle (1.9,1.9);
        \draw[black,thin,dashed] (1.1,2.1) rectangle (1.9,2.9);
        \draw[black,thin,dashed] (1.1,3.1) rectangle (1.9,3.9);
        \draw[black,thin,dashed] (2.1,0.1) rectangle (2.9,0.9);
        \draw[black,thin,dashed] (2.1,1.1) rectangle (2.9,1.9);
        \draw[black,thin,dashed] (2.1,2.1) rectangle (2.9,2.9);
        \draw[black,thin,dashed] (2.1,3.1) rectangle (2.9,3.9);
        \draw[black,thin,dashed] (3.1,0.1) rectangle (3.9,0.9);
        \draw[black,thin,dashed] (3.1,1.1) rectangle (3.9,1.9);
        \draw[black,thin,dashed] (3.1,2.1) rectangle (3.9,2.9);
        \draw[black,thin,dashed] (3.1,3.1) rectangle (3.9,3.9);
        \draw[black,thin,dashed] (4.1,0.1) rectangle (4.9,0.9);
        \draw[black,thin,dashed] (4.1,1.1) rectangle (4.9,1.9);
        \draw[black,thin,dashed] (4.1,2.1) rectangle (4.9,2.9);
        \draw[black,thin,dashed] (4.1,3.1) rectangle (4.9,3.9);

        \draw[red,very thick] (0.2,3.3) rectangle (0.5,3.7);
        \node at (0.7,3.5) {\Large\textcolor{red}{A}};
        \draw[blue,very thick] (1.4,3.2) rectangle (4.4,3.8);
        \node at (1.6,3.5) {\Large\textcolor{blue}{B}};
        \draw[purple,very thick] (0.3,0.7) rectangle (0.7,1.7);
        \node at (0.5,0.45) {\Large\textcolor{purple}{C}};
        \draw[green,very thick] (2.7,0.7) rectangle (4.8,3.5);
        \node at (4.5,1.5) {\Large\textcolor{green}{D}};
    \end{tikzpicture}
    \caption{This figure illustrates the  types of $4$-cycles, marked $A$, $B$, $C$, and $D$, in a quasi-cyclic parity-check matrix $H$ with circulant blocks pictured as the small squares.}
\end{figure}

Lemma~\ref{lem:4cycles} gives the number of type A $4$-cycles in any individual circulant submatrix $H_{i,j}$. Counting $4$-cycles between any pair of circulant block submatrices in the same row is essentially immediate from Lemma~\ref{lem:4cycles2}, but we will provide a restatement with more general notation as well as cover the case where block matrices are in the same column in the propositions that follow shortly. First, we establish notation for intersection numbers in this more general context.\\

\begin{definition}
\phantom{hi}
    \begin{enumerate}
        \item 
         Let $h_{i,j}, h_{i,j'} \in \F_2[x]/\langle x^r - 1 \rangle$ with $i \neq j$ be polynomials associated with circulant matrices $H_{i,j}, H_{i,j'} \in \F_2^{r \times r}$.
        For $k \in [r]_0$, define the \textit{$k^\text{th}$ column intersection number} to be
        \[ c_{i,j,j',k} := \left| \supp(h_{i,j}) \cap \supp(h_{i,j'}^k) \right|.\]
        \item 
         Let $h_{i,j}, h_{i',j} \in \F_2[x]/ \langle x^r - 1 \rangle$ with $i \neq i'$ be polynomials associated with circulant matrices $H_{i,j}, H_{i',j} \in \F_2^{r \times r}$.
        For $k \in [r]_0$, define the \textit{$k^\text{th}$ row intersection number} to be
        \[ r_{i,i',j,k} := \left| \supp(h_{i,j}) \cap \supp(h_{i',j}^k) \right|.\]
        \item Define the \textit{maximum row intersection} between the circulant matrices $H, H' \in \F_2^{r \times r}$ as
    \[\text{MRI}(H,H') = \max_{i,j \in [r] } \left| \supp \row_i H \cap \supp \row_j H' \right|.\]
    \end{enumerate}
\end{definition}

The number of 
Type B $4$-cycles between arbitrary circulants sharing a block row is given in the next result which has Lemma~\ref{lem:4cycles2} as a special case.\\

\begin{proposition} \label{prop:typeb}
    Let $h_{i,j}, h_{i,j'} \in \F_2[x]/ \langle x^r-1 \rangle$ with $j \neq j'$ be two polynomials corresponding to circulant submatrices in $H$ as in (\ref{eq:qcbigmatrix}). Let $H_{i,j}, H_{i,j'} \in \F_2^{r \times r}$ be the circulant matrices corresponding to $h_{i,j}$ and $h_{i,j'}$, respectively. Then 
    \begin{enumerate}
        \item $\text{MCI}(H_{i,j}, H_{i,j'}) = \max_{k \in [r_0]} c_{i,j,j',k}$ and
        \item the number of $4$-cycles between $H_{i,j}$ and $H_{i,j'}$ is
        \[ r \cdot \sum_{k=0}^{r-1} \binom{ c_{i,j,j',k}}{2}.\]
    \end{enumerate}
\end{proposition}

Type C $4$-cycles between circulants sharing a block column can be counted similarly.\\

\begin{proposition} \label{prop:typec}
    Let $h_{i,j}, h_{i',j} \in \F_2[x]/ \langle x^r-1 \rangle$ with $i \neq i'$ be two polynomials corresponding to circulant submatrices in $H$ as in (\ref{eq:qcbigmatrix}). Let $H_{i,j}, H_{i',j} \in \F_2^{r \times r}$ be the circulant matrices corresponding to $h_{i,j}$ and $h_{i',j}$, respectively. Then
    \begin{enumerate}
        \item $\text{MRI}(H_{i,j}, H_{i',j}) = \max_{k \in [r]_0} r_{i,i',j,k}$ and
        \item the number of $4$-cycles between $H_{i,j}$ and $H_{i',j}$ is
        \[ r \cdot \sum_{k=0}^{r-1} \binom{r_{i,i',j,k}}{2}.\]
    \end{enumerate}
\end{proposition}

\begin{proof}
    This follows from Proposition~\ref{prop:typeb} and Remark~\ref{rem:rowcoleq}.
\end{proof}

Now we develop the tools to count the number of type D $4$-cycles within $H$. \\

\begin{proposition} \label{prop:typed}
    Let $H \in \F_{2}^{rc \times r n_0}$ be a quasi-cyclic matrix as  in (\ref{eq:qcbigmatrix}). Let $h_{i,j} \in \F_2[x]/\langle x^r-1 \rangle$ for $i \in [c]_0$ and $j \in [n_0]_0$ be the polynomials corresponding to the circulant submatrices $H_{i,j}$ of $H$. The number of $4$-cycles in $H$ is given by the number of tuples 
 $(i, i', j, j', m, m', n, n') \in [c]_0^2 \times [n_0]_0^2 \times [r]_0^4$ with $i < i'$ and $j < j'$
    such that
    \begin{equation} \label{eq:typed}
        h_{i,j}[x^{n-m}] = h_{i,j'}[x^{n'-m}] = h_{i',j}[x^{n-m'}]=h_{i',j'}[x^{n'-m'}] = 1
    \end{equation}
    where exponent operations are taken modulo $r$. Further, the search space for such tuples can be reduced by setting any one of $n$, $n'$, $m$, or $m'$ to $0$ (or any other index) and multiplying the number of such tuples by $r$.
\end{proposition}

\begin{proof}
    We first show that any such tuple corresponds to a unique $4$-cycle. We claim that such a tuple corresponds to a $1$ in the following matrix entries: $(H_{i,j})_{m,n}$, $(H_{i,j'})_{m,n'}$, $(H_{i',j})_{m',n}$, and $(H_{i',j'})_{m',n'}$. We will show that $(H_{i,j})_{m,n} = 1$ if and only if $h_{i,j}[x^{n-m}]=1$. That $h_{i,j}$ corresponds to $H_{i,j}$ and vice versa is immediate from the correspondence between circulant matrices and polynomials. If $(H_{i,j})_{m,n} = 1$, there is some corresponding nonzero entry in $\row_1H_{i,j}$. Because the entry $(H_{i,j})_{m,n}$ is in $\row_m H_{i,j}$, we need to move $m$ indices to the left. This gives $\col_{n-m}H_{i,j}$ in $\row_1H_{i,j}$. Hence if $(H_{i,j})_{m,n} = 1$, it must be that $h_{i,j}[x^{n-m}] = 1$. The argument is similar in the other direction to show that $h_{i,j}[x^{n-m}] = 1$ implies $(H_{i,j})_{m,n} = 1$. The arguments for $(H_{i,j'})_{m,n'}$, $(H_{i',j})_{m',n}$, and $(H_{i',j'})_{m',n'}$ are also similar. For uniqueness, note that any $4$-cycle involving one entry in each of four different circulant sub-matrices has a unique set of two row indices ($i$, $i'$) and column indices ($j$, $j'$). The $i < i'$ and $j < j'$ conditions on the tuple guarantee that any $4$-cycle is only counted once.

    We will now show that any $4$-cycle corresponds to a unique tuple. A $4$-cycle involving four different circulants involves exactly two rows and two columns, where the two rows lie in two different block rows and the two columns lie in two different block columns. Let $i,i' \in [c]_0$ be the block row indices with $i < i'$, and $j, j' \in [n_0]_0$ be the block column indices with $j < j'$. The row index in block row $i$ is some value $m \in [r]_0$ from the top of the block row and similarly for $m'$, and the column index in block column $j$ is some value $n \in [r]_0$ from the leftmost column and similarly for $n'$. So any $4$-cycle corresponds to exactly one tuple.

    Now, let $(i,i',j,j',m,m',n,n')$ be a tuple corresponding to a $4$-cycle. Fix $i$, $i'$, $j$, and $j'$. If row $m$ in block row $i$, row $m'$ in block row $i'$, column $n$ in block row $j$, and column $n'$ in block row $j'$ all have $1$'s, then because of the circulant nature of the block matrices, going down one row and right one column also corresponds to a $4$-cycle, i.e., $(i, i', j, j', m+1, m'+1, n+1, n'+1)$ is also a $4$-cycle. So any of these four indices $m$, $m'$, $n$, and $n'$ can be fixed to some value and the total number with the three remaining indices varying can be multiplied by $r$.
\end{proof}

It remains an open question to obtain a closed form expression for the number of Type D $4$-cycles when $n_0>2$. Because these types of cycles do not arise in the parity-check matrices for the BIKE cryptosystem (which has $n_0=2$), we are able to count all $4$-cycles that may impact BIKE in Theorem~\ref{co:all4cycles}.

\section{Filtering Based on Cycle Structure} \label{section:filter}

In this section, we establish a filter that utilizes the tools developed in the previous section,  inspired by the filtering algorithm given in \cite{V21} which filters both Type II and Type III weak keys. Recall that Type I weak keys a subset of Type II weak keys.
In \cite{V21}, it was shown that Type I weak keys contribute more to DFR than Type II weak keys because they have larger column intersection values that are not their maximum column intersection value. 

We are now ready to present the new filtering algorithm, Algorithm~\ref{alg:filtering}. 
By keeping track of column intersection values in addition to the maximum column intersection, the procedure weights Type I weak keys more heavily. {It filters keys with other column intersection values, motivated by the following observation.}  
Suppose that $T$ is the maximum column intersection we would like the filter to avoid. An algorithm that avoids a $T$ column intersection is avoiding $r \binom{T}{2}$ $4$-cycles in $H$. If one would like to avoid the \rmv{number of} $4$-cycles caused by a single $T$ column intersection, then, intuitively, $4$-cycles that contribute to column intersections less than $T$ should be avoided also. The number of $4$-cycles in a key is given by
\[ r \cdot \left( \sum_{i=2}^{d-1} \gamma(i,h_0) \binom{i}{2} + \sum_{i=2}^{d-1} \gamma(i,h_1) \binom{i}{2} + \sum_{i=0}^{r-1} \binom{c_i}{2} \right)=
r \cdot \sum_{\delta = 2}^{d-1}\left( S_{0,\delta} + S_{1,\delta} + Q_\delta\right) \binom{\delta}{2}.\]
Here, the expression on the left-hand side is from  Theorem~\ref{co:all4cycles}, and the right-hand side is as used in Algorithm \ref{alg:filtering}.
This expression suggests a natural choice for a weight vector: $\mathbf{w}$, where $w_i = \binom{i}{2}$ gives the number of $4$-cycles corresponding to a pair of columns with intersection $i$. Notice that different choices for $\mathbf{w}$ could be selected, providing flexibility in the algorithm.

The differences between Algorithm \ref{alg:filtering} and that in \cite{V21}
are the following:
\begin{itemize}
    \item A weight vector $\mathbf{w}$ must be specified, and corresponds to the ``badness'' of any column intersection value.\\
    \item The threshold $s$ is no longer directly tied to maximum column intersection.\\
    \item Distance spectra ($S_0$, $S_1$, and $Q$ in Algorithm \ref{alg:filtering}) must be allocated to memory.\\
    \item What was the Type III weak key filter now checks column intersections between a column in $H_0$ and all columns in $H_1$. The existing algorithm computed $d^2$ vector weights in this area and this algorithm computes $r$ vector weights.\\
    \item There is an additional step before rejecting keys.\\
\end{itemize}
The only additional computation in the new algorithm is the step for rejecting keys. A key cannot be rejected until all spectra have been calculated. In the original algorithm, if a key failed the Type II test, the Type III test would not run. In this sense, Algorithm \ref{alg:filtering} requires only minimally more computational power. 

Algorithm \ref{alg:filtering} does not distinguish between weak key types and instead assumes that all $4$-cycles contribute equally to decoder failure. The algorithm could easily be modified so that there are different weights for column intersections between $H_0$ and $H_1$ and within one of $H_i$ could easily be implemented. 

The operation $\ast$ in Line 12 corresponds to computing the common nonzero coefficients between $h_0$ and $x^k h_1$, which are also the column intersection numbers.

\begin{algorithm}[ht]
\caption{Filter weak keys based on column intersection profile.}\label{alg:filtering}
\hspace*{\algorithmicindent} \textbf{Input:} Block size $r$, column weight $d$, weight vector $\mathbf{w} = (w_0, \dotsc, w_{d-1})$\\
\hspace*{\algorithmicindent} \phantom{\textbf{Input:}} threshold $s$, key $h_0 = x^{j_{0,1}} + \dots + x^{j_{0,d}}$, $h_1 = x^{j_{1,1}} + \dots + x^{j_{1,h}}$ 
\begin{algorithmic}[1]
\For{$i \in \{0,1\}$}
    \State $S_i \leftarrow 0$
    \For{$k \in \{1, \dotsc, d\}$}
        \For{$\ell \in \{k+1, \dotsc, d\}$}
            \State $\delta \leftarrow d(j_{i,k}, j_{i,\ell})$
            \State $S_{i,\delta} \leftarrow S_{i,\delta} + 1$
        \EndFor
    \EndFor
\EndFor
\State $Q \leftarrow 0$
\For{$k \in \{0, \dotsc, r-1\}$}
    \State $\delta \leftarrow \left| h_0 \ast x^k h_1 \right|$
    \State $Q_\delta \leftarrow Q_\delta +1$
\EndFor
\State $T \leftarrow 0$
\For{$\delta \in \{0, \dotsc, d-1\}$} 
    \State $T = w_\delta (S_{0,\delta} + S_{1,\delta}+ Q_\delta)$
    \If{$T \geq s$}
        \State \Return Reject
    \EndIf
\EndFor

\State \Return Accept
\end{algorithmic}
\end{algorithm}

Algorithm \ref{alg:filtering} filters keys with the $(m,\epsilon)$-gathering property.
A similar circumstance as Type I weak keys is almost certainly responsible for the class of weak keys with the $(m,\epsilon)$-gathering property introduced in \cite{WWW23} (see Definition~\ref{def:gatheringproperty}). These are keys $(\mathbf{h}_0,\mathbf{h}_1) \in \mathcal{K}_{m,\epsilon}$ such that most (all but $\epsilon$ many) of the nonzero entries of $\mathbf{h}_0$ are contained within a length $m$ interval of $[r]_0$. In \cite{WWW23}, $m < \frac{r}{2}$. Because these keys have their indices packed together, they are more likely to have higher multiplicities of column intersections, even if they do not have a large enough maximum column intersection to be filtered as Type II weak keys. 
Hence, Algorithm~\ref{alg:filtering} will also filter the $(m,\epsilon)$-gathering weak keys, or at least the subset of these keys that is likely to be detrimental in decoding.

\section{Complete Avoidance of $4$-Cycles} \label{section:ideal}

A natural question to ask, given that short cycles contribute so negatively to iterative decoding, is if $4$-cycles can be avoided in their entirety. In the BIKE setting, for both $h_0, h_1 \in \F_2[x]/\langle x^r-1\rangle$, $\mu(\delta,h_i) \in \{0,1\}$ for all $\delta \in \{1, \dotsc, \left\lfloor r/2 \right\rfloor\}$. In other words, there are no repeated distances in the distance spectra of $h_0$ and $h_1$. This means that each of the $\binom{d}{2}$ distances from $h_0$ are distinct, each of the $\binom{d}{2}$ distances from $h_1$ are distinct, and these two sets of distances are also distinct. This requires, at minimum, $\binom{d}{2} + \binom{d}{2} = 2 \binom{d}{2}$ different distances. But there are only $\left\lfloor \frac{r}{2} \right\rfloor$ possible distances, and so we require
\[ \left\lfloor \frac{r}{2} \right\rfloor \geq 2 \binom{d}{2}.\]
Rearranging and losing the floor function, this simplifies to $r > 2 d(d-1)$. More simply,
\[r > 2 d^2.\]
The MDPC code requirement for BIKE is that $2d \approx \sqrt{2r}$. This indicates that it should theoretically be possible to avoid $4$-cycles completely. However, none of the suggested BIKE parameters given in Table~\ref{table:BIKE_params} are such that $r > 2 d^2$. 

In \cite{V21}, the author approximates the probabilities of a random $h_0 \in \F_2[x]/\langle x^r - 1 \rangle$ will have maximum column intersection $m$. 
Assuming the independence of multiplicities of the spectrum, they show  
for $m \in [d]_0$ and $\delta \in \left[ \left\lfloor \frac{r}{2} \right\rfloor \right]$ and  $h \in \F_2[x]/\langle x^r -1 \rangle$ that: 
\begin{itemize}
    \item the probability of a particular multiplicity is
 $$P \left( \mu(\delta,h) = m \right) = \frac{r \binom{d-1}{d-m-1} \binom{r-d-1}{d-m-1}}{(d-m) \binom{r}{d}}.$$
        \item the probability of a particular upper bound on the multiplicity is $$P \left( \max_{\delta \in \{1, \dotsc, \left\lfloor r/2 \right\rfloor\}} \mu(\delta,h) < m \right)  = (1 - \pi_m)^{\left\lfloor r/2 \right\rfloor}. $$
    \end{itemize}

Table~\ref{table:ideals} gives some small values of $r$ as well as the $r$ suggested for BIKE with $\lambda=128$ bit security ($r = 12323$) and provides $p_{<2}$, the probability that a random $h_0 \in \F_2[x]/\langle x^r-1\rangle$ is $4$-cycle free. The suggested value for $d$ in the $r = 12323$ setting is $142$. The maximum $d$ value with $r=12323$ that could possibly have a $4$-cycle free representation is $d=78$. However, as shown in the table, with $r=12323$, for any value of $d > 20$, the probability of a polynomial having this property is essentially zero. This demonstrates that it is not feasible to try to obtain such keys randomly.

\begin{table*}[ht]
\begin{center}
\begin{tabular}{|l|c|c|c|}
% \toprule
\hline
$d$ & $r = 557$ & $r = 587$ & $r = 12323$ \\
\hline
$4$ & 0.9682 & 0.9698 & 0.9984 \\
$5$ & 0.8986 & 0.9035 & 0.9944 \\
$6$ & 0.7675 & 0.7778 & 0.9851 \\
$7$ & 0.5782 & 0.5941 & 0.9666 \\
$8$ & 0.3661 & 0.3844 & 0.9343 \\
$9$ & 0.1850 & 0.2005 & 0.8825 \\
$10$ & 0.0708 & 0.0801 & 0.8061 \\
$11$ & 0.0195 & 0.0233 & 0.7025 \\
$12$ & NA & 0.0047 & 0.5738 \\
$13$ & NA & NA & 0.4297 \\
$14$ & NA & NA & 0.2869 \\
$15$ & NA & NA & 0.1650 \\
$20$ & NA & NA & 0.0002 \\
\hline
\end{tabular} \\ \vspace{.2cm} 
\caption{The probability that a random $h_0$ with $r$, $d$ as given will have maximum column intersection $1$.}
 \label{table:ideals}
% \end{tiny}
\end{center}
\end{table*}

QC-MDPC codes without $4$-cycles can be constructed algebraically by using \textit{difference families}, which are a highly structured class of sets. This idea was previously considered for quasi-cyclic LDPC codes in \cite{PHNS13}, with a focus on high-rate codes with column weights $3$ or $4$. Specific algebraic code constructions from highly structured sets are unlikely to be helpful in cryptographic settings, since these sets are not large enough to behave as sets of random codes. Even so, they may be interesting to study as codes in their own right.

We close this section with a brief discussion on general QC-MDPC codes. In \cite{SA18}, the authors provide bounds on the weights of circulant matrices making up a quasi-cyclic parity-check matrix that could result in a quasi-cyclic matrix free of $4$-cycles. The result is stated  below, rephrased to better fit the setting in this paper. \\

\begin{remark}\label{thm:bounds_on_r} 
In \cite{SA18}, the authors provide bounds on the weights of circulant blocks of a quasi-cyclic parity-check matrix $H \in F_2^r$ that could result in a quasi-cyclic matrix free of $4$-cycles, showing that a necessary condition for $H$ to have girth (at least) $6$, the circulants must be of size $r \geq \max\{\mathcal{X}, \mathcal{Y}, \mathcal{Z}\}$ where 
$$
\begin{array}{ccl}
\mathcal{X} &=& \max_{i \in [c]_0} \left\lbrace 2 \sum_{j=0}^{n_0-1} \binom{\text{wt}(\mathbf{h}_{i,j})}{2}\right\}, \\ \ \\
\mathcal{Y} &=& \max_{j \in [n_0]_0} \left\lbrace 2 \sum_{i=0}^{c-1} \binom{\text{wt}(\mathbf{h}_{i,j})}{2} \right\rbrace,\\ \ \\
\mathcal{Z} &=& \max_{i \neq i', i, i' \in [c]_0} \left\lbrace \sum_{j=0}^{n_0-1} \left( \text{wt}(\mathbf{h}_{i,j}) \cdot \text{wt}(\mathbf{h}_{i',j} \right) \right\rbrace.
\end{array}
$$
 \end{remark}
  
The BIKE setting can be recovered from the more general Remark~\ref{thm:bounds_on_r} by taking $c=0$, $n_0=1$, and $\text{wt}(h_{i,j}) = d$. It is worth noting that taking $r$ larger than this bound does not guarantee that there exists any QC-MDPC code parity-check matrix with the given circulant row weight $d$ without $4$-cycles.

\section{Experimental results} \label{section:exp}

We computed the number of $4$-cycles in the recorded decoding failures of non-weak keys in  three cases, making use of the data provided by \cite{ABHLPR22} and available at \cite{ABHLPR22github}. In addition, we computed the number of $4$-cycles in the recorded decoding failures of weak keys in both the $r=557$ and $r=587$ cases with $T=3$ (no such data was provided for $T=4$ and $r=587$). The work of  \cite{ABHLPR22}  aimed to estimate the DFR of BIKE by calculating it explicitly for small values of $r$ and extrapolating from those the DFRs for values of $r$, $d$, and $t$ proposed for actual use. They considered values of $r$ prime with $389 \leq r \leq 827$ such that $x^r -1$ has only two irreducible factors modulo $2$ and use $d = 15$ and $t = 18$. We focused on the $r=557$ and $r=587$ cases, as they also provided weak key data for those two values. In both $r=557$ and $r=587$, they used $T=3$ as a threshold, which corresponds to filtering keys with any column intersections higher than $2$. In the $r=587$ case, they also used $T=4$ as a threshold, which filtered keys with column intersections higher than $3$.

We generated $1000$ keys in each of the three cases, using the same filtering algorithm as the weak keys that resulted in decoding failures used. Tables~\ref{tab:r557t3}, \ref{tab:r587t3}, and \ref{tab:r587t4} provide some basic differences in the number of $4$-cycles in both the keys that resulted in decoding failures and the randomly generated keys. In the tables, $H_i$ Average, Minimum, and Maximum are combined data from cycles within the single circulants $H_0$ and $H_1$, because there is no structural difference between cycles of these types. The numbers given for $H_0 \mid H_1$ Average, Minimum, and Maximum indicate $4$-cycles with a variable node in each of $H_0$ and $H_1$. The values for Total Average, Minimum, and Maximum are the sums of $4$-cycles in $H_0$, $H_1$, and between the two circulants.

\begin{table*}[ht]
\centering
\begin{tabular}{|l|c|c|}
\hline
& keys with decoding failures & randomly generated keys \\
\hline
$H_i$ Average & 7307 & 7103 \\
$H_i$ Minimum & 1114 & 1671 \\
$H_i$ Maximum & 13368 & 12811 \\
$H_0 \mid H_1$ Average & 18475 & 18032 \\
$H_0 \mid H_1$ Minimum & 13925 & 11697 \\
$H_0 \mid H_1$ Maximum & 23951 & 24058 \\
Total Average & 33090 & 32239 \\
Total Minimum & 26179 & 23394 \\
Total Maximum & 40104 & 41218 \\
\hline
\end{tabular}
\caption{Cycles in non-weak keys with decoding failures ($n=177$) and randomly generated non-weak keys ($n=1000$) when $r=557$ and $T=3$.}
\label{tab:r557t3}
\end{table*}

\begin{table*}[ht]
\centering
\begin{tabular}{|l|c|c|}
\hline
& keys with decoding failures & randomly generated keys \\
\hline
$H_i$ Average & 8307 & 7206 \\
$H_i$ Minimum & 1761 & 1761 \\
$H_i$ Maximum & 15849 & 14088 \\
$H_0 \mid H_1$ Average & 18472 & 18208 \\
$H_0 \mid H_1$ Minimum & 14088 & 11740 \\
$H_0 \mid H_1$ Maximum & 24067 & 26415 \\
Total Average & 35087 & 32620 \\
Total Minimum & 27002 & 23480 \\
Total Maximum & 44025 & 41677 \\
\hline
\end{tabular}
\caption{Cycles in non-weak keys with decoding failures ($n=128$) and randomly generated non-weak keys ($n=1000$) when $r=587$ and $T=3$.}
\label{tab:r587t3}
\end{table*}

\begin{table*}[ht]
\centering
\begin{tabular}{|l|c|c|}
\hline
& keys with decoding failures & randomly generated keys \\
\hline
$H_i$ Average & 10112 & 9328 \\
$H_i$ Minimum & 2348 & 1761 \\
$H_i$ Maximum & 20545 & 19958 \\
$H_0 \mid H_1$ Average & 21803 & 21438 \\
$H_0 \mid H_1$ Minimum & 15262 & 12914 \\
$H_0 \mid H_1$ Maximum & 28763 & 33459 \\
Total Average & 42026 & 40094 \\
Total Minimum & 33459 & 26415 \\
Total Maximum & 51656 & 59287 \\
\hline
\end{tabular}
\caption{Cycles in non-weak keys with decoding failures ($n=42$) and randomly generated non-weak keys ($n=1000$) when $r=587$ and $T=4$.}
\label{tab:r587t4}
\end{table*}

In order to determine if any of these differences were significant, we calculated the $p$-values for the differences in the means of each of these groups. These are provided in Table~\ref{tab:pvalues}. It is clear from the values that there are more total $4$-cycles in the keys with decoding failures than those without. This is to be expected given the negative impact of short cycles on iterative decoding.

The results very strongly suggest that there are significantly more $4$-cycles within the $H_i$ in the keys with decoding failures. The evidence for prevalence of $4$-cycles across the two circulants is less clear, but the datasets for the two $r=587$ cases are smaller than the dataset for the $r=557$ case. It is possible that with with more keys that exhibited decoder failure, a difference would also be seen in this category.

\begin{table*}
    \centering
    \begin{tabular}{|l|c|c|c|}
        \hline
        & $r=557$, $T=3$ & $r=587$, $T=3$ & $r=587$, $T=4$ \\
        \hline
        $H_i$ & $0.041$ & $<0.001$ & $0.012$ \\
        $H_0 \mid H_1$ & $0.013$ & $0.208$ & $0.471$ \\
        Total & $<0.001$ & $<0.001$ & $0.015$ \\
        \hline
    \end{tabular}
\caption{$p$-values for the difference in the means of the number of $4$-cycles in each of the three cases.}
\label{tab:pvalues}
\end{table*}

Nonetheless, we see that prevalence of $4$-cycles, especially those within a single circulant, is a distinguishing factor between keys that exhibited decoder failure and arbitrary keys. This suggests that one should prioritize minimizing short cycles in a single circulant over concerns about cycles arising across circulants.

\section{Conclusion} \label{section:conclusions}

In this paper, we considered the prevalence of $4$-cycles as structures leading to decoding failure for the iterative decoder used in BIKE. We defined the full spectrum to explain the cycle structure of the graphs. We provided a new filter for the allowable keys in BIKE. We extended these results from the codes used in BIKE to more general parity-check codes. We obtained experimental results suggesting that the impact of $4$-cycles within a single circulant is more detrimental than those across circulants. In additional study, it may be worthwhile to perform parameter analysis as in the new work \cite{Arpin_Markov}. We note that Round 4 of the NIST PQC Standardization process concluded following submission of this manuscript, with HQC being selected. 

\section*{Declarations}

The authors have no financial or proprietary interests in any material discussed in this article. The data generated and analyzed in this article is available at \cite{MM24github}.

\bibliography{bike.bib}

\begin{thebibliography}{10}

\bibitem{BIKE22}
N.~Aragon, P.~Barreto, S.~Bettaieb, L.~Bidoux, O.~Blazy, J.-C. Deneuville,
  P.~Gaborit, S.~Ghosh, S.~Gueron, T.~G{\"u}neysu, C.~A. Melchor, R.~Misoczki,
  E.~Persichetti, J.~Richter-Brockmann, N.~Sendrier, J.-P. Tillich, V.~Vasseur,
  and G.~Zemor.
\newblock {BIKE}: bit flipping key encapsulation.
\newblock 2022.

\bibitem{Arpin_Markov}
S.~Arpin, J.~B. Lau, R.~Perlner, and A.~Robinson.
\newblock Error floor prediction with Markov models for {QC}-{MDPC} codes.
\newblock {\em Cryptology ePrint Archive}, 2025.

\bibitem{ABHLPR22}
S.~Arpin, T.~R. Billingsley, D.~R. Hast, J.~B. Lau, R.~Perlner, and
  A.~Robinson.
\newblock A study of error floor behavior in {QC}-{MDPC} codes.
\newblock In {\em International Conference on Post-Quantum Cryptography}, pages
  89--103. Springer, 2022.

\bibitem{ABHLPR22github}
S.~Arpin, T.~R. Billingsley, J.~B. Lau, R.~Perlner, and A.~Robinson.
\newblock Raw data and decoder for the paper ``a study of error floor behavior in {QC-MDPC} codes''.
\newblock \url{https://github.com/HastD/BIKE-error-floor} Accessed 2024-10-09.

\bibitem{AYU20}
N.~Aydin, B.~Yildiz, and S.~Uludag.
\newblock A class of weak keys for the {QC-MDPC} cryptosystem.
\newblock In {\em 2020 Algebraic and Combinatorial Coding Theory}, pages 1--4.
  IEEE, 2020.
  
\bibitem{Baldi_in_place}
M.~Baldi, A.~Barenghi, F.~Chiaraluce, G.~Pelosi, and P.~Santini.
\newblock Analysis of in-place randomized bit-flipping decoders for the design of LDPC and MDPC code-based cryptosystems.
\newblock In {\em International Conference on E-Business and Telecommunications} pages 151--174, 2021.

\bibitem{BBCPS21}
M.~Baldi, A.~Barenghi, F.~Chiaraluce, G.~Pelosi, and P.~Santini.
\newblock Performance bounds for {QC-MDPC} codes decoders.
\newblock In {\em Code-Based Cryptography Workshop}, pages 95--122. Springer,
  2021.

\bibitem{Baldi_08}
M.~Baldi, M.~Bodrato, and F.~Chiaraluce.
\newblock A new analysis of the {McEliece} cryptosystem based on {QC-LDPC}
  codes.
\newblock In R.~Ostrovsky, R.~De~Prisco, and I.~Visconti, editors, {\em
  Security and Cryptography for Networks}, pages 246--262, Berlin, Heidelberg,
  2008. Springer Berlin Heidelberg.

\bibitem{Baldi_ISIT_07}
M.~Baldi and F.~Chiaraluce.
\newblock Cryptanalysis of a new instance of {McEliece} cryptosystem based on
  {QC-LDPC} codes.
\newblock In {\em IEEE International Symposium on Information Theory}, pages
  2591--2595, 2007.

\bibitem{Bardet_weak_16}
M.~Bardet, V.~Dragoi, J.-G. Luque, and A.~Otmani.
\newblock Weak keys for the quasi-cyclic {MDPC} public key encryption scheme.
\newblock In D.~Pointcheval, A.~Nitaj, and T.~Rachidi, editors, {\em Progress
  in Cryptology -- AFRICACRYPT 2016}, pages 346--367, Cham, 2016. Springer
  International Publishing.

\bibitem{BMvT_78}
E.~Berlekamp, R.~McEliece, and H.~van Tilborg.
\newblock On the inherent intractability of certain coding problems (corresp.).
\newblock {\em IEEE Transactions on Information Theory}, 24(3):384--386, 1978.

\bibitem{Grover_v_McEliece}
D.~J. Bernstein.
\newblock {G}rover vs. {M}c{E}liece.
\newblock In N.~Sendrier, editor, {\em Post-Quantum Cryptography}, pages
  73--80, Berlin, Heidelberg, 2010. Springer Berlin Heidelberg.

\bibitem{Dehghan_20}
A.~Dehghan and A.~H. Banihashemi.
\newblock Counting short cycles in bipartite graphs: {A} fast
  technique/algorithm and a hardness result.
\newblock {\em IEEE Transactions on Communications}, 68(3):1378--1390, 2020.

\bibitem{DPTRU02}
C.~Di, D.~Proietti, I.~E. Telatar, T.~J. Richardson, and R.~L. Urbanke.
\newblock Finite-length analysis of low-density parity-check codes on the
  binary erasure channel.
\newblock {\em IEEE Transactions on Information Theory}, 48(6):1570--1579,
  2002.

\bibitem{DZAWN09}
L.~Dolecek, Z.~Zhang, V.~Anantharam, M.~J. Wainwright, and B.~Nikolic.
\newblock Analysis of absorbing sets and fully absorbing sets of array-based
  {LDPC} codes.
\newblock {\em IEEE Transactions on Information Theory}, 56(1):181--201, 2009.

\bibitem{DGK20}
N.~Drucker, S.~Gueron, and D.~Kostic.
\newblock On constant-time {QC}-{MDPC} decoders with negligible failure rate.
\newblock In {\em Code-Based Cryptography Workshop}, pages 50--79. Springer,
  2020.

\bibitem{ETV99}
T.~Etzion, A.~Trachtenberg, and A.~Vardy.
\newblock Which codes have cycle-free {T}anner graphs?
\newblock {\em IEEE Transactions on Information Theory}, 45(6):2173--2181,
  1999.

\bibitem{F04}
M.~P. Fossorier.
\newblock Quasicyclic low-density parity-check codes from circulant permutation
  matrices.
\newblock {\em IEEE Transactions on Information Theory}, 50(8):1788--1793,
  2004.

\bibitem{GJS16}
Q.~Guo, T.~Johansson, and P.~Stankovski.
\newblock A key recovery attack on {MDPC} with {CCA} security using decoding
  errors.
\newblock In {\em Advances in Cryptology--ASIACRYPT 2016: 22nd International
  Conference on the Theory and Application of Cryptology and Information
  Security}, pages 789--815. Springer, 2016.

\bibitem{halford_06}
T.~R. Halford, K.~M. Chugg, and A.~J. Grant.
\newblock Which codes have 4-cycle-free {T}anner graphs?
\newblock In {\em IEEE International Symposium on Information Theory}, pages
  871--875, 2006.

\bibitem{KS07}
C.~A. Kelley and D.~Sridhara.
\newblock Pseudocodewords of {T}anner graphs.
\newblock {\em IEEE Transactions on Information Theory}, 53(11):4013--4038,
  2007.
  
\bibitem{MACKAY_Postol_03}
D.~MacKay and M.~Postol.
\newblock Weakness of Margulis and Ramanujan-Margulis low-density parity-check codes.
\newblock {\em Electronic Notes in Theoretical Computer Science}, 74:97--104, 2003.

\bibitem{Mackay1997near}
D.~J. MacKay and R.~M. Neal.
\newblock Near {S}hannon limit performance of low density parity check codes.
\newblock {\em Electronics letters}, 33(6):457--458, 1997.

\bibitem{MM24github}
G.~L. Matthews and E.~McMillon.
\newblock Raw data for the paper ``a combinatorial approach to avoiding weak keys in the {BIKE} cryptosystem''.
\newblock \url{https://github.
com/emcmillon/BIKE-cycles} Accessed 2024-10-14.

\bibitem{McEliece}
R.~J. {McEliece}.
\newblock {A public-key cryptosystem based on algebraic coding theory}.
\newblock {\em Deep Space Network Progress Report}, 44:114--116, Jan. 1978.

\bibitem{MTSB13}
R.~Misoczki, J.-P. Tillich, N.~Sendrier, and P.~S. Barreto.
\newblock {MDPC}-{M}c{E}liece: {N}ew {M}c{E}liece variants from moderate
  density parity-check codes.
\newblock In {\em IEEE International Symposium on Information Theory}, pages
  2069--2073, 2013.

\bibitem{Monico_LDPC_00}
C.~Monico, J.~Rosenthal, and A.~Shokrollahi.
\newblock Using low density parity check codes in the {McEliece} cryptosystem.
\newblock In {\em IEEE International Symposium on Information Theory}, pages
  215--, 2000.

\bibitem{NCMV12}
D.~V. Nguyen, S.~K. Chilappagari, M.~W. Marcellin, and B.~Vasic.
\newblock On the construction of structured {LDPC} codes free of small trapping
  sets.
\newblock {\em IEEE Transactions on Information Theory}, 58(4):2280--2302,
  2012.

\bibitem{NIEDERREITER_86}
H.~Niederreiter.
\newblock Knapsack-type cryptosystems and algebraic coding theory.
\newblock {\em Problems of Control and Information Theory}, 15(2):157--166,
  1986.

\bibitem{NSPZNGD23}
M.~R. Nosouhi, S.~W. Shah, L.~Pan, Y.~Zolotavkin, A.~Nanda, P.~Gauravaram, and
  R.~Doss.
\newblock Weak-key analysis for {BIKE} post-quantum key encapsulation
  mechanism.
\newblock {\em IEEE Transactions on Information Forensics and Security},
  18:2160--2174, 2023.

\bibitem{PHNS13}
H.~Park, S.~Hong, J.-S. No, and D.-J. Shin.
\newblock Construction of high-rate regular quasi-cyclic {LDPC} codes based on cyclic difference families.
\newblock {\em IEEE Transactions on Communications}, 61(8):3108--3113, 2013.

\bibitem{SA18}
M.-R. Sadeghi and F.~Amirzade.
\newblock Analytical lower bound on the lifting degree of multiple-edge {QC}-{LDPC} codes with girth 6.
\newblock {\em IEEE Communications Letters}, 22(8):1528--1531, 2018.

\bibitem{Santini}
P.~Santini, M.~Battaglioni, M.~Baldi, and F.~Chiaraluce.
\newblock Analysis of the error correction capability of {LDPC} and {MDPC} codes under parallel bit-flipping decoding and application to cryptography.
\newblock {\em IEEE Transactions on Communications}, 68(8):4648--4660, 2020.

\bibitem{SV20b}
N.~Sendrier and V.~Vasseur.
\newblock About low {DFR} for {QC-MDPC} decoding.
\newblock In {\em International Conference on Post-Quantum Cryptography}, pages
  20--34. Springer, 2020.

\bibitem{SV20}
N.~Sendrier and V.~Vasseur.
\newblock On the existence of weak keys for {QC-MDPC} decoding.
\newblock 2020.

\bibitem{Shor_94}
P.~Shor.
\newblock Algorithms for quantum computation: {D}iscrete logarithms and
  factoring.
\newblock In {\em Proceedings 35th Annual Symposium on Foundations of Computer
  Science}, pages 124--134, 1994.

\bibitem{Shor_97}
P.~W. Shor.
\newblock Polynomial-time algorithms for prime factorization and discrete
  logarithms on a quantum computer.
\newblock {\em SIAM Journal on Computing}, 26(5):1484--1509, 1997.

\bibitem{Tanner}
R.~Tanner.
\newblock A recursive approach to low complexity codes.
\newblock {\em IEEE Transactions on Information Theory}, 27(5):533--547, 1981.

\bibitem{T18}
J.-P. Tillich.
\newblock The decoding failure probability of {MDPC} codes.
\newblock In {\em IEEE International Symposium on Information Theory}, pages
  941--945. IEEE, 2018.

\bibitem{V21}
V.~Vasseur.
\newblock {\em Post-quantum cryptography: a study of the decoding of
  {QC}-{MDPC} codes}.
\newblock PhD thesis, Universit{\'e} de Paris, 2021.

\bibitem{WWW23}
T.~Wang, A.~Wang, and X.~Wang.
\newblock Exploring decryption failures of {BIKE}: New class of weak keys and key recovery attacks.
\newblock In {\em Annual International Cryptology Conference}, pages 70--100, 2023.

\bibitem{W96}
N.~Wiberg.
\newblock Codes and decoding on general graphs.
\newblock {\em Ph.D. Dissertation, Linkoping University, Linkoping}, 1996.

\end{thebibliography}
\bibliographystyle{ieeetr} 

\end{document}